\documentclass[12pt, reqno]{amsart}

\usepackage{graphics, stackrel}
\usepackage{amsmath, amssymb, amsthm}
\usepackage{graphicx}
\usepackage{verbatim}
\usepackage{amsfonts}
\usepackage{mathtools}

\usepackage{natbib}

\usepackage{enumitem}
\usepackage[T1]{fontenc}

\usepackage{caption}
\usepackage{subcaption}

\usepackage{booktabs}

\usepackage{fancyvrb}
\usepackage[usenames, dvipsnames, svgnames,table]{xcolor}
\usepackage{mdwlist}

\usepackage{enumitem}
\setlist[enumerate]{itemsep=2pt,topsep=3pt}
\setlist[itemize]{itemsep=2pt,topsep=3pt}
\setlist[enumerate,1]{label=(\alph*)}

\usepackage{mathrsfs}
\usepackage{bbm}

\usepackage[left=1.25in, right=1.25in, top=1.0in, bottom=1.15in, includehead, includefoot]{geometry}
\usepackage[citecolor=blue, colorlinks=true, urlcolor=blue, linkcolor=blue]{hyperref}

\renewcommand{\phi}{\varphi}
\renewcommand{\epsilon}{\varepsilon}

\renewcommand{\leq}{\leqslant}
\renewcommand{\geq}{\geqslant}

\usepackage[ruled, linesnumbered]{algorithm2e}

\newcommand{\st}{\ensuremath{\ \mathrm{s.t.}\ }}
\newcommand{\setntn}[2]{ \{ #1 : #2 \} }

\newcommand{\1}{\mathbbm 1}

\newcommand{\given}{\,|\,}

\newcommand{\RR}{\mathbbm R}

\newcommand{\NN}{\mathbbm N}

\newcommand{\JJ}{\mathbbm J}

\theoremstyle{plain}
\newtheorem{theorem}{Theorem}[section]

\newtheorem{lemma}[theorem]{Lemma}

\theoremstyle{definition}

\newtheorem{remark}{Remark}[section]

\usepackage[normalem]{ulem}

\begin{document}


\title{}

\begin{center}
  \Large{Systemic Risk in Financial Systems: Properties of Equilibria}

    \vspace{1em}

  \large
  John Stachurski\footnote{Email: \texttt{john.stachurski@anu.edu.au}.  The
  author thanks Shu Hu for research assistance.
  Financial support from ARC grant FT160100423 is gratefully acknowledged.}

  \par \bigskip

  \small
  Australian National University \bigskip

  \normalsize
  \today
\end{center}

\begin{abstract}
    \vspace{1em}
Eisenbern and Noe (2001) analyze systemic risk for financial institutions linked by a network of liabilities.  They show that the solution to their model is unique when the financial system is satisfies a regularity condition involving risk orbits.  We show that this condition is not needed: a unique solution always exists.  

    \vspace{1em}
\end{abstract}

\maketitle

\section{Introduction}

\cite{eisenberg2001systemic} analyzed systemic
risk for institutions linked by a complex network of liabilities.  The
fundamental importance of their work was brought to the fore by the financial crisis of
2007--2008, where the solvency of individual firms became unclear due to
interlocking financial obligations. In such settings, 
defaults can trigger negative
feedback loops, where restructuring compromises the balance sheets of other
firms in the network, restricting their ability to make payments to debt
holders, and forces more rounds of restructuring. 

Because of these cyclical interactions in the networks,
obtaining a system of payments that clears the market is a fixed point
problem.  While existence of a fixed
point is easy to show, uniqueness is 
challenging.  \cite{eisenberg2001systemic} show that uniqueness holds when the
financial system is ``regular.''  This condition has become common in the
literature (see, e.g., \cite{cifuentes2005liquidity} or
\cite{feinstein2018sensitivity}).
We prove it is not needed: the limited liability and
priority conditions in \cite{eisenberg2001systemic} are themselves enough to uniquely identify an
equilibrium clearing vector.\footnote{The
    only exception is an extreme case where no
firm in the entire network has any cash at all.  This case is discussed in
Remark~\ref{r:extr} below.}

\section{Results}

If $a = (a(i))_{i=1}^n$ and $b = (b(i))_{i=1}^n$ are vectors in
$\RR^n$, then $a \leq b$ means $a(i) \leq b(i)$ for all $i$.  We write $a \ll
b$ if $a(i) < b(i)$ for all $i$.  The symbol $[a,b]$ represents the order
interval $\setntn{x \in \RR^n}{a \leq x \leq b}$.  A Markov matrix is a
nonnegative square matrix with unit row sums.  For a given Markov matrix
$\Pi$, we call $j$ \emph{accessible} from $i$ under $\Pi$ if
either $j=i$ or $\Pi^k(i, j) > 0$ for some $k \in \NN$.

As in \cite{eisenberg2001systemic}, a \emph{financial system} is a set
of nodes $I := \{1, 2, \ldots, n\}$, an $n \times n$ matrix $\Pi = (\Pi(i,
j))$ of relative liabilities, a vector $\bar p = (\bar p (i)) \in \RR^n_+$ of
nominal obligations ($\bar p(i)$ is the sum of all nominal obligations held by
node $i$) and a vector $e = (e(i)) \in \RR^n_+$ of external cash flows.  The
matrix of relative liabilities is a Markov matrix.  Following
\cite{eisenberg2001systemic}, we assume that $\bar p \gg 0$, so that, for any
given node, the total sum of liabilities to other nodes is not
zero.\footnote{This is required for $\Pi$ to have unit row sums.  See p.~239
of \cite{eisenberg2001systemic}.}

A \emph{clearing vector}
for a financial system $(I, \Pi, \bar p, e)$ is a vector of payments $p \in [0, \bar p]$
satisfying the limited liability restriction
\begin{equation}\label{eq:limli}
    p(j) \leq \sum_{i \in I} p(i) \Pi(i, j) + e(j)
\end{equation}
and the absolute priority condition
\begin{equation}\label{eq:ap}
    p(j) = \sum_{i \in I} p(i) \Pi(i, j) + e(j)
    \quad \text{ or }  \quad
    p(j) = \bar p(j)
\end{equation}
for all $j \in I$.  Combining these two restrictions and writing them in
vector form (cf.\ \cite{eisenberg2001systemic}, p.~240), the set of clearing
payment vectors is seen to coincide with the set of fixed points of the
mapping $p \mapsto \Phi p$ on $[0, \bar p]$ defined by 
\begin{equation}\label{eq:defphi}
    \Phi p := (p \Pi + e) \wedge \bar p.
\end{equation}
In~\eqref{eq:defphi} and below, all $n$-vectors are treated as row vectors.

\begin{remark}\label{r:extr}
    Limited liability and absolute priority cannot by themselves pin down
    outcomes for the extreme case where every firm in the network has zero
    operating cash.\footnote{For example, if $\psi$ is a stationary
        distribution for $\Pi$, $\lambda$ is a constant in $[0,1]$, $p =
    \lambda \psi$ and $\bar p = \psi$, then $e=0$ implies $\Phi p = (\lambda
    \psi \Pi + e) \wedge \bar p = (\lambda \psi) \wedge \psi = \lambda \psi = p$.
Since $\lambda$ was arbitrary in $[0,1]$, there is a continuum of equilibria.}
In what follows, we adopt the convention $p=0$ when $e=0$. This is natural if
firms cannot raise new loans to meet their liabilities when $e=0$,
    since a payment sequence cannot be initiated without outside capital.
\end{remark}

In \cite{eisenberg2001systemic}, the \emph{risk orbit} of a given node $i$ is
the set of all nodes that are accessible from $i$.  The financial system is
called \emph{regular} if the risk orbit of every node contains at least one
$j$ with $e(j)>0$.  A system where regularity fails is shown below.  Suppose
$e(1)=1$ and $e(2)= e(3) =0$.  Arrows represent nonzero liabilities.  The risk
orbit from node
2 is $\{2, 3\}$. Regularity fails because $e(2) =e(3) = 0$.  

\begin{figure}[h]
    \centering
    \scalebox{0.85}{\includegraphics[clip=true, trim=0mm 20mm 0mm 20mm]{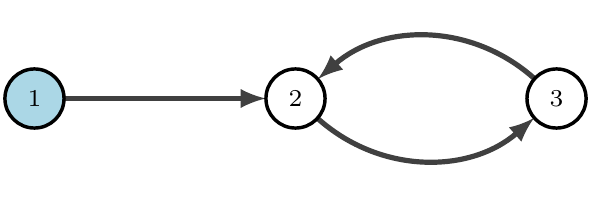}}
\end{figure}

Fortunately, regularity is irrelevant for uniqueness, 
  as the next theorem shows.

\begin{theorem}\label{t:clearbk}
    Every financial system has exactly one clearing payment vector.
\end{theorem}

To prove Theorem~\ref{t:clearbk}, we begin with a lemma.
In the lemma, we say that node $j$ in a financial system $(I, \Pi, \bar p, e)$,
is \emph{cash accessible} if there exists an $i \in I$ such that
$e(i) > 0$ and $j$ is accessible from $i$.

\begin{lemma}\label{l:cfpc}
    If every node in $S = (I, \Pi, \bar p, e)$ is cash accessible, then $S$ has
    a unique clearing vector $p^*$.  Moreover,  $p^* \gg 0$ and
    $\Phi^k p \to p^*$ as $k \to \infty$ when $0 \leq p \leq \bar p$.
\end{lemma}

\begin{proof}
    Let $S$ be as described. By the fixed point theorem in the appendix (Theorem~\ref{t:duext}),
    it suffices to show that $\Phi$ is an increasing concave self-map on
    $[0,\bar p]$ with $\Phi^k 0 \gg 0$ for some $k \in \NN$.
    As confirmed in \cite{eisenberg2001systemic}, $\Phi$ is increasing and concave,
    so only the last statement needs to be verified.  To this end, we set
    \begin{equation*}
        \delta 
        := \frac{1}{n^2} 
        \cdot \min 
        \left\{
            \setntn{\bar p(i)}{i \in I} \cup \setntn{e(i)}{i \in I \, \st e(i) > 0}
        \right\}.
    \end{equation*}
    Let $\hat e$ be defined by $\hat e(i) = 1$ if $e(i) > 0$ and zero
    otherwise.  We claim that, for all $m \leq n$, 
    \begin{equation}\label{eq:phikprop}
        \Phi^m 0 
        \geq \delta (\hat e + \hat e \Pi + \cdots + \hat e \Pi^{m-1}).
    \end{equation}
    This holds at $m=1$ because $\Phi 0 = e
    \wedge \bar p \geq \delta \hat e$.  Now suppose \eqref{eq:phikprop} holds at some $m \leq n
    - 1$.  Then, since $\Phi$ is increasing, we obtain
    \begin{align*}
        \Phi^{m+1} 0 
        & \geq (\delta (\hat e + \hat e \Pi + \cdots + \hat e \Pi^{m-1}) \Pi + e) \wedge \bar p
        \\
        & \geq (\delta (\hat e + \hat e \Pi + \cdots + \hat e \Pi^m) ) \wedge \bar p
    \end{align*}
    Since $\hat e + \hat e \Pi + \cdots + \hat e \Pi^m \leq n^2 \1$, where
    $\1$ is a vector of ones, and since $(\delta n^2 \1) \leq \bar p$
    by the definition of $\delta$, we have 
    $\Phi^{m+1} 0 \geq \delta (\hat e + \hat e \Pi + \cdots + \hat e \Pi^m)$.
    This argument confirms that \eqref{eq:phikprop} holds for all $m \leq n$.

    We now claim that $\Phi^n 0 \gg 0$.  In view of~\eqref{eq:phikprop}, it suffices
    to show that, for any $j \in I$, there exists a $k < n$ with 
        $(\hat e \Pi^k) (j) 
        = \sum_{i \in I} \hat e(i) \Pi^k(i, j) > 0$.
    Since every node in $S$ is cash accessible, we know there exists an $i
    \in I$ with $e(i) > 0$ and $j$ is accessible from $i$.  For this $i$ we can 
    choose $k \in \NN$ with $k < n$ and $\Pi^k(i, j) =  \hat e(i) \Pi^k(i, j) > 0$.  
    We conclude that $\Phi^n 0 \gg 0$, as claimed.
\end{proof}

In what follows, a subset $\JJ$ of $I$ is called \emph{absorbing} if no
element of its complement $\JJ^c := \setntn{i \in I}{i \notin \JJ}$ is
accessible from $\JJ$.  Also, for a given vector $v$ on $I$ and some $\JJ
\subset I$, we write $v \given \JJ$ for the restriction of $v$ to $\JJ$.
For matrix $M$ on $I \times I$ we write $M \given \JJ$ for the restriction
of $M$ to $\JJ \times \JJ$. 

\begin{proof}[Proof of Theorem~\ref{t:clearbk}]
    As pointed out by \cite{eisenberg2001systemic}, the operator $\Phi$ is
    increasing and concave on $[0, \bar p]$.  It follows from the
    increasing property and Tarski's fixed point theorem that at least one
    clearing vector always exists.\footnote{In fact $\Phi$ is a continuous
        self-map on the convex compact set $[0, \bar p]$, so 
    Brouwer's fixed point theorem gives the same conclusion.}
    The remainder of the proof focuses on uniqueness.

    Let $P$ be the set of all nodes in $I$
    that are cash accessible.   Let $A$ be all $i$ in $P^c$ such that
    $P$ is accessible from $i$.  Let $N$ be all $i$ in $P^c$ such that
    $P$ is not accessible from $i$.  Note that $I = P \cup A \cup
    N$ and that these sets are disjoint.

    The set $N$ is an absorbing set, since, by definition, $P$ is not
    accessible from $N$, and  $A$ cannot be accessible because otherwise
    $P$ would also be accessible.  The set $P$ is also absorbing because, if $j \in P^c$ is
    accessible from some $i \in P$, then $j$ is cash accessible.  But
    then $j \in P$, which is a contradiction.

    Iterating $k$ times on \eqref{eq:limli} gives
    \begin{equation}\label{eq:iterlimli}
        p \leq e + e \Pi + e \Pi^2 + \cdots + e \Pi^{k-1} + p \Pi^k .
    \end{equation}
    If $j$ is in $A$, then $j$ is not cash accessible, so $e(j) = 0$ and
    $(e \Pi^m) (j)  = 0$ for all $m$.  
    Hence~\eqref{eq:iterlimli} reduces to 
    \begin{equation}\label{eq:pjtoz}
        p(j) 
        \leq (p \Pi^k)(j) 
        \leq (\bar p \Pi^k)(j) 
        = \sum_{i \in I} \bar p(i) \Pi^k(i, j)
        = \sum_{i \in A} \bar p(i) \Pi^k(i, j).
    \end{equation}
    The last equality uses the fact that both $P$ and $N$ are absorbing
    sets.  Since $P$ is accessible from every element of $A$, when $i, j
    \in A$ we have $\Pi^k(i, j) \to 0$ as $k \to \infty$.\footnote{A Markov
        chain started at $i \in A$ leaves for $P$ with $\epsilon > 0$
        probability every $n = |I|$ periods.  Since $P$ is absorbing it
        never returns.  Hence the probability that
        the chain hits $j \in A$ after $k$ periods converges to zero with
        $k$.}   By taking $k$
    large in \eqref{eq:pjtoz}, we see that $p(j) = 0$ for all $j \in A$.

    Since $N$ is absorbing, $\Pi \given N$ is a Markov matrix and 
    $(N, \Pi \given N, \bar p \given N, e \given N)$ is itself an
    independent financial system.\footnote{Although nodes in $N$ may have
        inbound links from $A$, we have just shown that payments from
        $A$ are zero.  At the same time, there are no inbound links from
    $P$, since $P$ is aborbing.}
        Moreover, $P$ contains all cash accessible nodes so no
    element of $N$ is cash accessible and, in particular, $e \given N =
    0$. Hence $p \given N = 0$ by the convention in Remark~\ref{r:extr}.

    It remains only to treat nodes in $P$.  
    Since $P$ is absorbing, $\Pi \given P$ is a Markov matrix and 
    $(P, \Pi \given P, \bar p \given P, e \given P)$ is also an
    independent financial system.\footnote{Again, while nodes in $P$ may have
        inbound links from $A$, we have just shown that payments from
        $A$ are zero.}  Moreover, every node in $P$ is cash accessible,
        so, by Lemma~\ref{l:cfpc}, a unique clearing vector $p^* \gg
        0$ exists on $P$.  After extending to all nodes by setting $p^*(i) =
        0$ for all $i \notin P$, we have a unique clearing vector.
\end{proof}

\appendix

\section{Remaining Proofs}

In what follows, a self-map $F$ from a
subset $Y$ of $\RR^d$ to itself is called \emph{globally stable} if $F$ has a
unique fixed point $\bar y$ in $Y$ and $F^m y \to \bar y$ as $m \to \infty$
for all $y \in Y$.  We use the following fixed point theorem, which slightly
modifies Theorem~3.1 of~\cite{du1990}.  (See also Corollary~2.1.1. of
\cite{zhang2012variational}.)

\begin{theorem}\label{t:duext}
    Let $A$ be an increasing self-map on $[0, b] \subset \RR^d$.
    If $A$ is concave and there exists an $n \in \NN$ such that $A^n 0 \gg  0$,
    then $A$ is globally stable on $[0,b]$.
\end{theorem}

\begin{proof}
    Theorem~3.1 of~\cite{du1990} implies that any increasing concave
    self-map $F$ on $[0, b]$ satisfying $F 0 \gg 0$ is globally stable.
    Since compositions of increasing concave operators are increasing and
    concave, this implies that $A^n$ is globally stable on $[0, b]$.
    Denote its fixed point by $\bar v$.
    Since $\{ A^m 0 \}_{m \in \NN}$ is monotone increasing and since the subsequence
    $\{A^{mn} 0\}_{m \in \NN}$ converges up to $\bar v$ as $m \to \infty$, we must have
    $A^m 0 \to \bar v$.  A similar argument gives $A^m b \to \bar
    v$.  For any $v \in [0, b]$ we have $A^m 0 \leq A^m v \leq A^m b$, so $A^m v
    \to \bar v$ as $m \to \infty$.

    The last step is to show that $\bar v$ is the unique fixed point of $A$.
    From Tarski's fixed point theorem we know that at least one fixed point
    exists.  Now suppose $v \in [0, b]$ is such a point.  Then $v = A^m v$ for
    all $m$.  At the same time, $A^m v \to \bar v$ by the results just
    established.  Hence $v = \bar v$.  The proof is now complete.
\end{proof}

\bibliographystyle{ecta}

\bibliography{fr}

\end{document}